\documentclass[11pt]{article}

\usepackage{odonnell2}

\newcommand{\onote}[1]{\footnote{{\bf \color{blue}Ryan}: {#1}}}
\newcommand{\rnote}[1]{\footnote{{\bf \color{red}Rocco}: {#1}}}

\newcommand{\err}{\nu}
\newcommand{\thedist}{\calD}
\newcommand{\thedistsym}{\calD^{\textnormal{sym}}}
\newcommand{\estim}{\wh{\calD}}
\newcommand{\generalNoise}{\textnormal{Noise}}
\newcommand{\erasure}{\textnormal{Erase}}
\newcommand{\flip}{\textnormal{Flip}}
\newcommand{\zints}[1]{[0..#1]}

\makeatletter
\newenvironment{proofof}[1]{\par
  \pushQED{\qed}%
  \normalfont \topsep6\p@\@plus6\p@\relax
  \trivlist
  \item[\hskip\labelsep
\emph{    Proof of #1\@addpunct{.}}]\ignorespaces
}{%
  \popQED\endtrivlist\@endpefalse
}
\makeatother
\begin{document}

\title{Sharp bounds for population recovery}

\author{
Anindya De\\
Northwestern University\\
{\tt de.anindya@gmail.com}
\and
Ryan O'Donnell\thanks{Supported by NSF grant CCF-1618679.}\\
Carnegie Mellon University\\
{\tt odonnell@cs.cmu.edu}
\and
\and Rocco A.~Servedio\thanks{Supported by NSF grants CCF-1420349 and CCF-1563155.}\\
Columbia University \\
{\tt rocco@cs.columbia.edu}
}

\maketitle

\begin{abstract}
The \emph{population recovery problem} is a basic problem in noisy unsupervised learning  that has attracted significant research attention in recent years~\cite{WY12,DRWY12, MS13, BIMP13, LZ15,DST16}. A number of different variants of this problem have been studied, often under assumptions on the unknown distribution (such as that it has restricted support size).  In this work we study the sample complexity and algorithmic complexity of the most general version of the problem, under both bit-flip noise and erasure noise model.  We give essentially matching upper and lower sample complexity bounds for both noise models, and efficient algorithms matching these sample complexity bounds up to polynomial factors.

\end{abstract}

\section{Introduction}  \label{sec:intro}

\subsection{The erasure noise and bit-flip noise population recovery problems}

The \emph{noisy population recovery (NPR)} problem is to learn an unknown probability distribution $\thedist$ on $\{0,1\}^n$, under $\err$-noise, to $\ell_\infty$-accuracy~$\eps$.\footnote{With high probability.  Because we are not concerned with logarithmic factors in our time/sample complexity, we will for simplicity omit discussion of the standard tricks (independent repetition, taking the median of estimators) used to boost success probabilities. We will also always assume, without loss of generality, that $\eps$ is at most some sufficiently small absolute constant. }  In this problem the learner gets access to independent \emph{samples}~$\by$, each distributed as follows:  First $\bx \sim \thedist$, and then $\by \sim \generalNoise_\err(\bx)$, where $\generalNoise_\err(\cdot)$ denotes either the application of bit-flip noise or erasure noise (described below).  The learner's task is to output an estimate $\estim$ of $\thedist$ satisfying $\|\estim - \thedist\|_\infty \leq \eps$ (with high probability). For the sake of a compact representation, we assume the learner only outputs the nonzero values of~$\estim$; this means that a successful learner need only output $O(1/\eps)$ nonzero values.  
We are interested in minimizing both the \emph{sample complexity} and the \emph{running time} of learning algorithms.

A simpler variation of the NPR problem is the \emph{estimation} task.  Here the algorithm doesn't need to output a complete~$\estim$; it only needs to output an $\eps$-accurate estimate of~$\thedist(u)$ for a given input $u \in \{0,1\}^n$.  Certainly the estimation task is no harder than full NPR; conversely, it is known and not hard (see Section~\ref{sec:reduc}) that given the ability to do estimation, one can do full NPR with just a $\poly(n,1/\eps)$ factor slowdown.  Hence we mainly focus on estimation in this paper.

As mentioned above, we consider two different models of noise.  Each involves a parameter $0 < \err < 1$; smaller values of~$\err$ correspond to more noise, so $\err$ may be better thought of as a ``correlation'' parameter.

\paragraph{Erasure noise.} For $x \in \{0,1\}^n$ we define $\erasure_{1-\err}(x)$ to be the distribution on $\{0,1,?\}^n$ given by independently replacing each coordinate of~$x$ with the symbol `$?$' with probability $1-\err$.  Thus~$\err$ is the retention probability for each coordinate.

\paragraph{Bit-flip noise.}  For $x \in \{0,1\}^n$ we define $\flip_{\frac{1-\err}{2}}(x)$ to be the distribution on $\{0,1\}^n$ given by independently flipping each coordinate of~$x$ with probability $\frac{1-\err}{2}$.  Equivalently, each coordinate of $x$ is retained with probability $\err$ (as in erasure noise), and is otherwise replaced with a uniformly random bit.  This is also the model of noise associated to the so-called ``Bonami--Beckner noise operator''~$T_\err$.

\subsection{Our results}

For the bit-flip noise population recovery problem, our main result is a lower bound on the sample complexity of estimation, as well as a full NPR algorithm whose running time (hence also sample complexity) matches it up to polynomial factors:

%
%
%
%

\begin{theorem} \label{thm:main-bit-flip}
   Let $\eps > 0$ be sufficiently small and let $n \in \N$.  Then any estimation algorithm for NPR with bit-flip noise must use at least the following number of samples:
   \[
    \begin{cases}
        \exp\left(\Theta\left(n^{1/3} \cdot \ln^{2/3}(1/\epsilon)/\err^{2/3} \right)\right) & \quad \text{if } \frac{\ln(1/\eps)}{n} \leq \err \leq 1/2, \\
        \exp\left(\Theta\left(n^{1/3} \cdot \ln^{2/3}(1/\epsilon) \cdot (1-\err)^{1/3}\right)\right) & \quad\text{if } 1/2 \leq \err \leq 1-\frac{\ln(1/\eps)}{n}. \\
    \end{cases}
    \]
    Furthermore, there is an algorithm for the full NPR problem with bit-flip noise having running time and samples equal to the above times $\poly(n,1/\eps)$.
\end{theorem}

Prior to this work and the very recent and independent work of \cite{PTW17}, no nontrivial upper or lower bounds were known even for the sample complexity of the general bit-flip noise population recovery problem.  (See \cite{WY12,LZ15,DST16} for earlier works that gave upper bounds and algorithms under the additional assumption that the unknown distribution $\thedist$ is guaranteed to be supported on at most $k$ strings.)

For the erasure noise population recovery problem, our main results are also essentially matching upper and lower bounds on sample complexity and an algorithm whose running time is polynomial in $n$ and the sample complexity:

\begin{theorem} \label{thm:main-erasure}
Let $\eps > 0$ be sufficiently small and let $n \in \N$.  

\begin{itemize}
\item Assume that $\sqrt{16\ln(1/\eps)/n} \leq \err \leq 1/160$.
 Then any estimation algorithm for NPR with erasure noise  must use at least $1/\epsilon^{\Omega(1/\nu)}$ samples.
    \item  There is an algorithm for the full NPR problem with erasure noise using time and samples at most $\poly(n, 1/\eps^{1/\err})$.
\end{itemize}
\end{theorem}

For this problem, in earlier work \cite{MS13} gave an algorithm with sample complexity and running time $(n/\eps)^{O(\log(1/\err)/\err)}$.

Finally, we note that  in very recent and independent work, \cite{PTW17} have obtained very similar results to Theorems~\ref{thm:main-bit-flip} and~\ref{thm:main-erasure} for the population recovery problem.

\subsection{Our techniques}

Our approach is similar in spirit to, and shares some technical similarities with, the recent work of \cite{DOS16,NP16} on the trace reconstruction problem.  We  take an analytic view on the combinatorial process defined by the bit-flip and erasure noise operators, and convert the sample complexity questions for these population recovery problems to questions about the extrema of real-coefficient polynomials satisfying certain conditions on various circles in the complex plane; we then obtain our sample complexity bounds by analyzing these extremal polynomial questions. The main algorithmic ingredient in our results is linear programming.

\section{Prelimaries}

\subsection{Well-known preliminary reductions} \label{sec:reduc}

\paragraph{Estimation, enumeration, and recovery.}  Variants of the NPR problem with relaxed goals have been studied in the literature.  One is the aforementioned \emph{estimation} problem.  Another (complementary) variant is called \emph{enumeration}:  in the enumeration problem, the learning algorithm is only required to output a list of strings $x_1, \dots, x_m$ that is guaranteed (with high probability) to include all strings that have probability at least~$\eps$ under~$\calD$; such strings are sometimes referred to as ``heavy hitters.''  Batman et al.~\cite{BIMP13} gave a range of results for the enumeration problem.

It is easy to see that a solution to the estimation problem can be efficiently bootstrapped to full NPR given the ability to solve the enumeration problem (simply run estimation, with a sufficiently boosted success probability, on each of the $m$ strings in the list obtained from enumeration).  In turn, it is also well known that an estimation algorithm can be efficiently transformed into an enumeration algorithm  via a ``branch-and-prune'' approach.  Roughly speaking, such an approach maintains a not-too-large (size at most $O(1/\eps)$) set of $i$-bit prefixes that is known to contain all the ``heavy hitters''; to construct the set of $(i+1)$-bit prefixes, the approach first ``branches'' to extend each $i$-bit prefix $x$ to both $x0$ and $x1$, and then ``prunes'' any element of $\{x0,x1\}$ that is determined, using the estimation procedure, not to be a heavy hitter.  (Note that since only heavy hitters are maintained it will again be the case that the set of $(i+1)$-bit prefixes has size at most $O(1/\eps)$.)  As \cite{BIMP13} observe, an early example of such a branch-and-prune routine that performs enumeration given an oracle for estimation is the Goldreich--Levin algorithm~\cite{GoldreichLevin:89} for list-decoding the Hadamard code.  Both Dvir et al.~\cite{DRWY12} and Batman et al.~\cite{BIMP13} give fairly detailed analyses of the above-described reduction from enumeration to estimation; we omit the details here and refer the interested reader to Section~6.1 of \cite{DRWY12} and Section~2 of~\cite{BIMP13} respectively.

Summarizing the reductions discussed above, we have that NPR is (up to polynomial factors) no harder than the estimation problem, and it is also clearly no easier than estimation (since estimation is a subproblem of general NPR).  Thus in the rest of this paper we restrict our attention to the estimation problem.

\ignore{Given such a list, $m$ invocations of the estimation problem (with sufficiently boosted success probability) suffice to solve the full NPR problem.  As an example use of this strategy, we may rely on the following theorem of Batman~et~al.:
\begin{theorem}                                     \label{thm:batman}
    \onote{cite BIMP}
    In the bit-flip model of noise, the enumeration problem can be solved using $n^{O(\log(1/\eps)/\err^2)}$ samples and time.\onote{This comes from Batman's first ``simple'' pruning algorithm.  It already takes a little reading to determine the dependence on $\err$ here.  I wasn't able to figure out the $\err$-dependence in their ``quasipoly-in-$\eps$'', more complicated second pruning algorithm.  Note that $n^{O(\log(1/\eps)/\err^2)}$ is more or less okay for us, although it's not negligible once $\log(1/\eps)/\err$ starts getting close to $1/n^{\Theta(1)}$.  Have to work this all out more carefully at some point, or else avoid using BIMP as a black box.}
\end{theorem}
}

\paragraph{Symmetrization.}  We further recall some well-known tricks that have been used in  past papers on NPR.  First, in the estimation problem, we may assume without loss of generality that the string~$u$ whose probability is to be estimated is $u = (0, \dots, 0)$.  This is because the learner can easily convert samples from $\thedist$ to samples from ``$\thedist \oplus u$''.

Next, for the problem of estimating $\thedist(0, \dots, 0)$, we may assume without loss of generality that $\thedist$ is \emph{symmetric}, meaning that it gives equal probability mass to all strings at the same Hamming weight.  In other words, $\thedist$ is effectively given by a probability distribution $\thedistsym$ on $\zints{n}$, with $\thedist(x) = \thedistsym(|x|)/\binom{n}{|x|}$.  On one hand, if $\thedist(0, \dots, 0)$ can be estimated in the general case, it can certainly be estimated in the symmetric case.  On the other hand, given a general distribution $\thedist$, the learner can randomly permute the coordinates of each sample, effectively obtaining access to samples from a symmetric distribution $\thedistsym$, with $\thedistsym(0, \dots, 0) = \thedist(0, \dots, 0)$.  Thus it suffices for the learner to be able to estimate in the symmetric case.

In this symmetric case, we will write the unknown $\thedistsym$ more simply as a probability (row) vector $[p_0\ p_1\ \cdots \ p_n]$.  Although the learner observes full strings, it may as well only consider the Hamming weights of the strings it receives.  Thus we may think of it as obtaining samples from the probability (row) vector $[q_0\ q_1\ \cdots\ q_n]$, where
\begin{equation} \label{eqn:qpA}
    q = pA, \qquad A_{ij} = \Pr[\text{a weight $i$ string becomes a weight $j$ string under $\err$ noise}].
\end{equation}
It is not hard to write down the entries of $A$ in either noise model.  We remark that, after symmetrization, the bit-flip model becomes equivalent to running the well-known \emph{Ehrenfest urn model} for continuous time $t n$, where $e^{-t} = \err$.  It is easy to write down the known generating function for that model:
\begin{proposition}   [\cite{Siegert47,BellmanHarris51}]                          \label{prop:gen-fcn-flip}
    For $A$ associated to the $\flip_{\frac{1-\err}{2}}$ noise model, and $z$ an indeterminate,
    \[
        \sum_{j=0}^n A_{ij} z^j = \left(\frac{1-\err}{2} + \frac{1+\err}{2}z\right)^i \left(\frac{1+\err}{2} + \frac{1-\err}{2}z\right)^{n - i}.
    \]
\end{proposition}
For the erasure model, the generating function is even simpler.  The following is easily verified:
\begin{proposition}                                     \label{prop:gen-fcn-erase}
    For $A$ associated to the $\erasure_{1-\err}$ noise model, and $z$ an indeterminate,
    \[
        \sum_{j=0}^n A_{ij} z^j = \bigl((1-\err) + \err z\bigr)^i.
    \]
\end{proposition}

To recap, in the estimation problem the learner's task is to estimate $p_0$ to accuracy~$\eps$, given samples from~$q$.  We recall the well-known fact that, by taking the empirical distribution of $O(n/\delta^2)$ samples, the learner may obtain an estimate $\wh{q}$ of $q$ satisfying $\|\wh{q} - q\|_1 \leq \delta$ (with high probability).  Although $q = pA$, as noted in previous works one unfortunately cannot effectively estimate $p_0$ simply as the first coordinate of $\wh{q}A^{-1}$, because $A$ is very poorly conditioned.  Instead one needs a more sophisticated approach.

\section{Reduction to an analytic problem}

It is not hard to characterize the optimal sample complexity for the estimation problem.  Define
\[
    \eta(\eps,\err) = \min_{\substack{\text{probability vectors } p, p'  \\ |p_0 - p'_0| > 2\eps}} \|pA - p'A\|_1
\]
(where the parameter $\err$ implicitly appears within~$A$). If two probability vectors $p$ and $p'$ have $|p_0 - p'_0| > 2\eps$, then a successful estimation algorithm must be able to distinguish the two cases.  But if $q = pA$, $q' = p'A$ are close, in the sense that $\|q - q'\|_1 \leq \delta$, then a learning algorithm will need $\Omega(1/\delta)$ samples to distinguish them with high probability.  We conclude:
\begin{proposition} \label{prop:sc-lower}
    The sample complexity of any population recovery algorithm --- indeed, any estimation algorithm --- is $\Omega(1/\eta(\eps,\err))$.
\end{proposition}
On the other hand, suppose the lower bound $\eta(\eps,\err) \geq \delta$ holds.  Consider an estimation algorithm that first produces an empirical estimate $\wh{q}$ with $\|\wh{q} - q\|_1 < \delta$ using $O(n/{\delta}^2)$ samples, and then exactly solves the following optimization problem using linear programming:
\[
    \min_{\text{probability vectors } p'} \| \wh{q} - p'A\|_1.
\]
(This can be efficiently written as an LP with $O(n)$ variables and constraints and with rational numbers of  $\poly(n)$ bit-complexity.\footnote{For simplicity in this paper we assume that $\eps$ and $\err$ are rational quantities of $\poly(n)$ bits known to the learning algorithm.})  We claim that any optimal solution $p'$ will have $|p_0 - p'_0|  \leq 2\eps$.  Otherwise, by definition $\|pA - p'A\|_1 > \eta(\eps,\err) \geq \delta$; but $\|\wh{q} - pA\|_1 < \delta$, a contradiction.  Thus we get an efficient solution to the estimation problem (except with accuracy only $2\eps$).  In conclusion, we have established the following:
\begin{proposition} \label{prop:upper}
    The estimation problem can be solved with $\poly(n, 1/\eta(\eps/2,\err))$ time and samples.
\end{proposition}
\noindent Thus we see that, up to polynomial factors, both the sample complexity and runtime complexity of the estimation problem is effectively controlled by the parameter $\eta(\eps,\err)$.

We now further simplify the definition of $\eta(\eps,\err)$, similar to what was done in \cite{DOS16}. The difference of two probability vectors is precisely any vector in the set
\[
    \Delta = \{[c_0\ c_1\ \cdots\ c_n] : \littlesum_i c_i = 0,\ \littlesum_i |c_i| \leq 2\}.
\]
Thus we have that
\[
    \eta(\eps,\err) = \min_{\substack{c \in \Delta\\ c_0 > 2\eps}} \|cA\|_1.
\]
Furthermore, elementary complex analysis (see e.g. Proposition 3.5 of \cite{DOS16}) shows that, for complex~$z$ and column vector $\mathfrak{z} = (1, z, z^2, \dots, z^n)$,
\[
    \max_{|z| = 1} |cA\mathfrak{z}| \leq \|cA\|_1 \leq \sqrt{n+1} \cdot \max_{|z| = 1} |cA\mathfrak{z}|.
\]
Note also that $cA\mathfrak{z}$ is a polynomial in~$z$ that is easily calculated from the generating function of the noise process (see Propositions~\ref{prop:gen-fcn-flip},~\ref{prop:gen-fcn-erase}). We obtain:
\begin{theorem}                                     \label{thm:polys}
    Up to a factor of~$\sqrt{n+1}$ on the outside, and up to a factor of~$2$ on $\eps$, we have that
    \[
        \eta(\eps,\err) \asymp \min_{\substack{c \in \Delta\\ c_0 > 2\eps}}
                                        \begin{cases}
                                            \max_{|z| = 1} |F_c(z)| & \text{in the $\flip_{\frac{1-\err}{2}}$ noise model,} \\
                                            \max_{|z| = 1} |E_c(z)| & \text{in the $\erasure_{1-\err}$ noise model,}
                                         \end{cases}
    \]
    where
    \begin{align}
            F_c(z) &= \sum_{i=0}^n c_i \left(\frac{1-\err}{2} + \frac{1+\err}{2}z\right)^i \left(\frac{1+\err}{2} + \frac{1-\err}{2}z\right)^{n - i},
            \label{eq:Fc}\\
            E_c(z)  &= \sum_{i=0}^n c_i \bigl((1-\err) + \err z\bigr)^i. \label{eq:Ec}
    \end{align}
\end{theorem}
Given $c \in \Delta$ with $c_0 > 2\eps$, define the following polynomial (with real coefficients and a complex parameter):
\[
    Q_c(v) = \sum_{i=0}^n c_i v^i.
\]
Thus the assumptions on~$c$ are equivalent to $Q_c(0) > 2\eps$, $Q_c(1) = 0$, and $L(Q_c) \leq 2$, where $L(Q_c)$ is the \emph{length} of $Q_c$; i.e., the sum of the absolute values of its coefficients.

In analyzing $E_c$ above, we use that $E_c(z) = \sum_{i=0}^n c_i u^i$, where $u = (1-\err) + \err z$.  As $z$ traces out the unit circle $|z| = 1$, the parameter $u$ traces out the circle $\bdry D_\err(1-\err)$ of radius $\err$ centered at the real value $1-\err$.  Thus
\[
    \max_{|z|=1} |E_c(z)| = \max_{u \in \bdry D_\err(1-\err)} |Q_c(u)|,
\]
where
\[
    Q_c(v) = \sum_{i=0}^n c_i v^i.
\]

In analyzing $F_c$ above, we use that
\begin{equation} \label{eq:Fczw}
    F_c(z) 
    = \Bigl(\tfrac{\err}{\frac{1+\err}{2} - \frac{1-\err}{2}w}\Bigr)^n \sum_{i=0}^n c_i w^i, \qquad \text{where } w = \frac{\frac{1-\err}{2} + \frac{1+\err}{2}z}{\frac{1+\err}{2} + \frac{1-\err}{2}z}.
\end{equation}
The parameter $w$ (being a M\"{o}bius transformation of~$z$) traces out the unit circle as~$z$ does, and for $w = e^{i\theta}$ it is not hard to compute that
\begin{equation} \label{eq:burger}
    \left|\tfrac{\err}{\frac{1+\err}{2} - \frac{1-\err}{2}w}\right|^2 = \frac{2\err^2}{(1-\cos \theta) + (1+\cos \theta)\err^2}
    =  \frac{1}{1 + \frac{(1-\err^2) \sin^2(\theta/2) }{\err^2}}.
\end{equation}
Thus
\begin{equation} \label{eq:FcQ}
    \max_{|z|=1} |F_c(z)| = \max_{-\pi < \theta \leq \pi}   \bigg( \frac{1}{1 + \frac{(1-\err^2) \sin^2(\theta/2) }{\err^2}}\bigg)
^{n/2} \cdot |Q_c(e^{i\theta})|.
\end{equation}
We finally conclude:
\begin{corollary}                                     \label{cor:polys}
    Up to a factor of~$\sqrt{n+1}$ on the outside, and up to a factor of~$2$ on $\eps$, we have that
    \[
        \eta(\eps,\err) \asymp \min_{Q}
                                        \begin{cases}
                                            \displaystyle \max_{-\pi < \theta \leq \pi}   \bigg( \frac{1}{1 + \frac{(1-\err^2) \sin^2(\theta/2) }{\err^2}}\bigg)^{n/2} \cdot |Q(e^{i\theta})|  & \text{in the $\flip_{\frac{1-\err}{2}}$ noise model,} \\
                                            \displaystyle \max_{u \in \bdry D_\err(1-\err)} |Q(u)|, & \text{in the $\erasure_{1-\err}$ noise model,}
                                         \end{cases}
    \]
    where the minimum is over real-coefficient polynomials~$Q$ of degree at most~$n$ satisfying $Q(0) > 2\eps$, $Q(1) = 0$, and $L(Q) \leq 2$.
\end{corollary}

Combining Propositions~\ref{prop:sc-lower} and~\ref{prop:upper} with Corollary~\ref{cor:polys}, we see that
Theorems~\ref{thm:main-bit-flip} and~\ref{thm:main-erasure} follow from giving bounds on the two quantities specified in Corollary~\ref{cor:polys} (or in Theorem~\ref{thm:polys}).  We give such bounds in the following sections.

\section{Circle bounds for erasure noise}

\subsection{A lower bound on $ \eta(\eps,\err)$ for erasure noise}
Notice that $L(Q) \leq 2$ implies that $|Q(u)| \leq 2$ for all $|u| = 1$.  We have the following:

\begin{theorem}
    Let $Q$ be a complex polynomial with $|Q(0)| \geq 2\eps$ and $|Q(u)| \leq 2$ for $|u| = 1$.  Then for $0 < \err < 1/2$ we have $\max_{u \in \bdry D_\err(1-\err)} |Q(u)| \geq 2\eps^{\frac{1-\err}{\err}}$.  (For $1/2 \leq \err \leq 1$ we immediately get a lower bound of $2\eps$, by the Maximum Modulus Principle.)
\end{theorem}
\begin{proof}
    Let $U$ be the unit circle, let $O$ be the circle of radius $1/2$ centered at $1/2$, which lies inside $U$, and let $C = \bdry D_\err(1-\err)$, which lies inside $O$.   The M\"{o}bius transformation $A(u)= 1/(1-u)$ takes these circles to vertical lines $U'$, $O'$, and $C'$ with real parts $1/2$, $1$, and $1/2\err$, respectively.  Defining the function $f(z)=Q(A^{-1}(z))$, we have that $f$ is bounded on the strip defined by $U'$ and $C'$, and we have that $\sup_{y \in U'} |f(y)| \leq 2$, $\sup_{y \in O'} |f(y)| \geq 2\eps$.\ignore{\rnote{Replaced here ``We know that $|Q| \leq 2$ on $U'$ and $|Q| \geq 2\eps$ on $O'$.''}}  Writing $M$ for the maximum modulus of $f$ on $C'$, the Hadamard Three-Lines Theorem implies that
    \[
        2^{\frac{1-2\err}{1-\err}} M^{\frac{\err}{1-\err}} \geq 2\eps,
    \]
    which completes the proof after rearrangement.
\end{proof}

\ignore{
\begin{verbatim}
For any $z$ lying in the strip bounded by the lines $U'$ and $C'$, the point $A^{-1}(z)$ lies on a circle that
passes through 1 and is centered on the real line at some point in the
interval [0,1-\err] (hence the circle is contained in the unit circle)

Consider the function f(z) = Q(A^{-1}(z)).  This is a bounded function of
z=x+iy defined on the strip bounded by U' and C' (as on any such input
point z, the value of f is Q(some point in the unit circle).  It is this function
f for which we know that |f| \leq 2 on U' and \sup_{y on the line O'} |f(y)|
\geq 2\eps, and it's this function that we apply the H3L theorem to (see
e.g. the theorem statement at https://en.wikipedia.org/wiki/Hadamard_three-lines_theorem ).
\end{verbatim}
}

\subsection{An upper bound on $ \eta(\eps,\err)$ for erasure noise}

In this section, we will prove the following theorem:
\begin{theorem}~\label{thm:erasure-upper-bound}
There is an absolute constant $\tau>0 $ such that for every $\nu \le 1/10$, $0<\eps<\tau$ and $\ln(1/\epsilon)/\nu^2 \leq n$, there exists a vector $c \in \Delta$ with $c_0 > 2\epsilon$ such that the polynomial $Q_c(v) = \sum_{i=0}^n c_i v^i$ satisfies
\[
\sup_{v \in \partial D_{\nu/16}(1-\nu/16)} |Q_c(v)|  = \epsilon^{-\Omega(1/\nu)}.
\]
\end{theorem}
In order to prove this theorem, we will collect a few facts at the beginning.
Given $a,r>0$, define
the set $B_{a,r}$ as
\[
B_{a,r} = \big\{ (1-8a) + 4a(z + z^{-1}) \ : z \in \partial D_r(0)\big\}.
\]
We now make a few observations about the set $B_{a,r}$ as $r$ varies.
In particular, we have the following fact:
\begin{fact}~\label{fact:ellipse-1}
For $r \in \{1,2,4\}$, the sets $B_{a,r}$ are as follows:
\begin{itemize}
\item For $r=1$, the set $B_{a,r}$ is the line segment joining $1$ and $1-16 a$.
\item For $r=2$, the set $B_{a,r}$ is the ellipse centered at $1-8a$ with major axis
$[1-8a-10a, 1-8a + 10a]$ and minor axis $[1-8a + 6i, 1-8a-6i]$.
\item For $r=4$, the set $B_{a,r}$ is the ellipse centered at $1-8a$ with major axis
is  $[1-8a -17a , 1-8a  + 17a]$ and minor axis is $[1-8a +15i, 1-8a -15i]$.
\end{itemize}
\end{fact}
It is quite easy to observe that the circle $D_{4a}(1-4a)$ is contained in $B_{a,2}$. By Hadamard's three circle theorem, any holomorphic function $f$ satisfies
\begin{equation}~\label{eq:three-circ-bound}
\sup_{z \in D_{4a}(1-4a)} |f(z)| \le \sup_{z \in B_{a,2}} |f(z)| \le \sqrt{\sup_{z \in B_{a,1}} |f(z)| } \cdot \sqrt{\sup_{z \in B_{a,4}} |f(z)| }.
\end{equation}
Consequently, we have the following corollary.
\begin{corollary}~\label{corr:supremum-poly-circle}
Let $c \in \Delta$ and $Q_c(v) = \sum_{i=0}^n c_i v^i$. Then,
$$
\sup_{z \in D_{4a}(1-4a)} |Q_c(z)| \le \sqrt{\sup_{z \in B_{a,1}} |Q_c(z)| } \cdot 2 \sqrt{\exp(9an)}.
$$
\end{corollary}
\begin{proof}
We apply (\ref{eq:three-circ-bound}) to the function $Q_c$ and then observe that  $$\sup_{z \in B_{a,4}} |Q_c(z)| \le \sup_{z \in B_{a,4}}  |z|^n \cdot (\sum_{j=0}^n |c_j|) \le 2 \cdot (1 +9a)^n \le 2 \cdot \exp(9an),
$$
which concludes the proof.
\end{proof}

We  next recall the following theorem from \cite{Erdelyi16}:

\begin{theorem} [Lemma~3.3 of \cite{Erdelyi16}] \label{thm:Erdelyi-extremal}
For any $L \in [0,1/17)$ and $M \in \mathbb{N}$, there is a real-coefficient polynomial $p(z) = \sum_{j=0}^{M} a_j z^j$  with $|a_0| \ge L \cdot (\sum_{j=1}^M |a_j|)$ such that $p$ has at least
$ T_{L,M}= \min \{\frac{2}{7} \sqrt{M \cdot (-\ln L)}, M\}$ repeated roots at $1$.
\end{theorem}

We will also use the following result from \cite{BEK:97}:
\begin{claim} [Lemma~5.4 of \cite{BEK:97}] \label{clm:upper-bound-line}
Let $p: \mathbb{C} \rightarrow \mathbb{C}$ be defined as $p(x) = \sum_{j=0}^M a_j z^j$ where $|a_j| \le 1$ for all $0 \le j \le n$. Further, let $p$ have $k$ repeated roots at $1$. Let $A$ define the interval $[1-k/(9M), 1]$. Then
$$
\sup_{z \in A} |p(z)| \le (M+1) \bigg( \frac{e}{9}\bigg)^k.
$$
\end{claim}
With these two results in hand, we are now ready to prove Theorem~\ref{thm:erasure-upper-bound}.
\begin{proofof}{Theorem~\ref{thm:erasure-upper-bound}}
Let us set $M  = \lfloor \ln(1/\epsilon)/\nu^2 \rfloor$ and let $p(z)  = \sum_{j=0}^M c_j z^j$  be the polynomial from Theorem~\ref{thm:Erdelyi-extremal} with $L = 2\epsilon$. Let us also scale the coefficients such that $|c_0| = 2\epsilon$ and thus $\sum_{j=0}^M |c_j| \leq 2$. As $\ln(1/\epsilon)/\nu^2 \le n$, $M \le n$ and thus our construction is well-defined. The polynomial $p$ has at least $T$ roots at $1$, where
$$
T = \min \bigg\{\frac{2}{7} \frac{\ln(1/\epsilon)}{\nu}, \frac{\ln(1/\epsilon)}{\nu^2}\bigg\} = \frac{2}{7} \frac{\ln(1/\epsilon)}{\nu}.
$$
Let us define $\theta = T/(9M)  = (2/63) \cdot \nu$.
By applying Claim~\ref{clm:upper-bound-line}, it follows that
$$\sup_{[1-\theta, 1]} |p(z)| \le  (M+1) \cdot \bigg( \frac{e}{9}\bigg)^T \le \bigg(\frac13\bigg)^T.$$
Here the last inequality uses the relation between $T$ and $M$ and $\epsilon \le \tau$. Finally, set $a = \nu/63$. Then, applying Corollary~\ref{corr:supremum-poly-circle}, we obtain
$$
\sup_{z \in D_{4a}(1-4a)} |p(z)| \le \sqrt{\sup_{z \in B_{a,1}} |p(z)| } \cdot 2 \sqrt{\exp(9aM)} \le \sqrt{\bigg( \frac{1}{3} \bigg)^T \cdot 4 \cdot  \exp(9aM)}.
$$
Plugging in $a = \nu/63$ and $T = (2 M \nu)/7$, we obtain that
$$
\sup_{z \in D_{4a}(1-4a)} |p(z)|  \le (1/\epsilon)^{\Omega(\nu)},
$$
which concludes the proof.
\end{proofof}

\section{Circle bounds for bit-flip noise}

\subsection{A lower  bound on $ \eta(\eps,\err)$ for bit-flip noise}

In this section we prove the following theorem:

\begin{theorem} \label{thm:bitflip-pos}
For $0 < \err,\eps < 1$ and $n \in \N$ which satisfy
${\frac {2 \ln(2/\eps)} n} \leq \err \leq 1 - {\frac {2 \ln(2/\eps)} n},$
we have
\[
\eta(\epsilon, \err) \geq\eps \cdot \exp \bigg( - O\bigg(\frac{\ln^{2/3}(1/\epsilon) \cdot (n(1-\err^2))^{1/3}}{\err^{2/3}}\bigg)\bigg).
\]
\end{theorem}

\begin{proof}
Fix any vector  $[c_0 \ c_1 \ldots  c_n] \in \Delta$ with $|c_0| >2\epsilon.$
Recalling Theorem~\ref{thm:polys} and  (\ref{eq:Fczw}), to prove Theorem~\ref{thm:bitflip-pos} it suffices to show that the function $F_c(z)$ as defined in (\ref{eq:Fczw}) satisfies
\begin{equation} \label{eq:Fcgoal-lb}
\max_{|z| =1} |F_c(z)| \geq \eps \cdot \exp \bigg( -O \bigg(\frac{\ln^{2/3}(1/\epsilon) \cdot (n(1-\err^2))^{1/3}}{\err^{2/3}}\bigg)\bigg).
\end{equation}
To prove this, we recall (\ref{eq:FcQ}) which states that 
$$
 \max_{|z|=1} |F_c(z)| = \max_{-\pi < \theta \leq \pi}   \bigg( \frac{1}{1 + \frac{(1-\err^2) \sin^2(\theta/2) }{\err^2}}\bigg)
^{n/2} \cdot |Q_c(e^{i\theta})|.
$$
Next, we observe that  for $-\pi < \theta \leq \pi$, we have
\[
 1 - (1-\err^2) \sin^2(\theta/2) \in \left[1- {\frac {(1-\err^2) \theta^2} 4},1- {\frac {(1-\err^2) \theta^2}{16}}\right],
\]
where the last inclusion uses $\theta^2/16 \leq \sin^2(\theta/2) \leq \theta^2/4$, which holds for $\theta \in [-\pi,\pi]$.   Using the elementary fact $e^{-x} \le 1/(1+x)$ for all $x \ge 0$, it follows that
$$
\bigg( \frac{1}{1 + \frac{(1-\err^2) \sin^2(\theta/2) }{\err^2}}\bigg) \ge \exp \left( -\frac{1-\nu^2}{4\nu^2} \theta^2\right)
$$ 
and thus, we have 
\begin{equation}\label{eq:Fcz-lower1}
\max_{|z|=1} |F_c(z)| \ge  \max_{-\pi < \theta \leq \pi}  \exp \left( -\frac{1-\nu^2}{8\nu^2} \theta^2n\right) \cdot |Q_c(e^{i \theta})|.
\end{equation}
To finish the proof, we 
recall Corollary~3.2 of \cite{BE:97}:
\begin{theorem}[Corollary~3.2 of \cite{BE:97}] There is a universal constant $c>0$ such that the following holds:
Let $Q(z)$ be a univariate polynomial with complex coefficients, $Q(z)=\sum_{j=0}^n b_j z^j$ with $|b_0|=1$ and all coefficients $|b_j| \leq M.$  Let $A$ be a subarc of the unit circle with length $a$, where $0 < a < 2\pi$.  Then there is some $w \in A$ such that
\[
|Q(w)| \geq \exp \left({\frac {-c(1+\ln M)}{a}}\right).
\]
\end{theorem}
Applying this theorem to the polynomial $Q_c$  with its ``$M$'' set to $1/c_0 $ and its ``$a$'' set to $\theta^\ast$ and combining with (\ref{eq:Fcz-lower1}), we obtain 
$$
\max_{|z|=1} |F_c(z)| \ge  \max_{-\pi < \theta^\ast \leq \pi}  \exp \left( -\frac{1-\nu^2}{8\nu^2} \theta^{\ast 2}n\right) \cdot \exp \left(-\theta(1) \cdot \frac{1 + \ln (1/c_0)}{\theta^\ast} \right). 
$$
Finally set $\theta^{\ast}$ as
$$
\theta^\ast = \frac{1}{10} \cdot \frac{\nu^{2/3} \cdot \ln^{1/3}(1/\epsilon) }{(n(1-\nu^2))^{1/3}}, a
$$
and plug in the right hand side of the above expression (it is easy to see that the constraints on $\nu$ imply that $\theta^\ast \le 1$). This finishes the proof.  

\end{proof}

\subsection{An upper bound on $ \eta(\eps,\err)$ for bit-flip noise}

In this section we prove the following theorem:
\begin{theorem}~\label{thm:lower-bound}
There is a universal constant $c>0$ such that for $\err, 0< \epsilon < c$ and $n \in \mathbb{N}$ which satisfy
$ \big(\frac{2 \ln(2/\epsilon)}{n}\big)^{1/4} \le \err \le 1 -\frac{2 \ln(2/\epsilon)}{n}$, we have
\[
\eta(\epsilon, \err) = \exp \bigg( - \Omega \bigg(\frac{\ln^{2/3}(1/\epsilon) \cdot (n(1-\err^2))^{1/3}}{\err^{2/3}}\bigg)\bigg).
\]
\end{theorem}

Recalling (\ref{eq:FcQ}), to prove this result we must demonstrate the existence of a vector $[c_0 \ c_1 \ldots  c_n] \in \Delta$, $|c_0| >2\epsilon$ such that $F_c(z)$ satisfies
\begin{equation} \label{eq:Fcgoal-ub}
\sup_{\Vert z \Vert=1} |F_c(z)| = \exp \bigg( \Omega \bigg(- \frac{\ln^{2/3}(1/\epsilon) \cdot (n(1-\err^2))^{1/3}}{\err^{2/3}}\bigg)\bigg),
\end{equation}
where we recall from Equation (\ref{eq:Fc}) that
 \[
F_c(z) = \sum_{i=0}^n c_i \left(\frac{1-\err}{2} + \frac{1+\err}{2}z\right)^i \left(\frac{1+\err}{2} + \frac{1-\err}{2}z\right)^{n - i}.
\]

To prove this, we will use Theorem~\ref{thm:Erdelyi-extremal} and
the following lemma, which relates the multiplicity of roots of a polynomial at $1$
with the supremum of $p$ on an arc centered at $1$.

\begin{lemma} [Lemma~4.7 in \cite{BE:97}] \label{lem:bound-repeated}
Suppose $p: \mathbb{C} \rightarrow \mathbb{C}$ is a polynomial of the form $p(z) = \sum_{j=0}^M a_j z^j$, where $|a_j| \le 9$ and $p$ has $k$ repeated roots at $1$. If $A$ denotes the arc of the unit circle that is symmetric around $1$ and has length $(2k)/(9M)$,
then
\[
\sup_{z \in A} |p(z)| \le 9(M+1) \cdot \bigg( \frac{e}{9}\bigg)^k.
\]
\end{lemma}
\begin{proofof}{Theorem~\ref{thm:lower-bound}}
With these results in hand we are ready to specify our construction of $[c_0, \ldots  c_n]$. For this, we set $M$ as follows:
$$
M = \lfloor n^{2/3} \cdot \ln^{1/3} (1/\epsilon) \cdot (1-\err^2)^{2/3} \cdot \err^{-4/3} \rfloor .
$$
We first make the following observations about $M$. (i) Since $\err^4 \ge \frac{\ln(1/\epsilon)}{n}$, it is the case that $M \le n$. (ii) Since $1- \err \ge 2 \ln (2/\epsilon)/n$, it is moreover the case that $M \ge \ln (1/\epsilon)$.

For $M$ as defined above,
let us rescale the polynomial in Theorem~\ref{thm:Erdelyi-extremal} so that $|a_0| = 2\epsilon$ and thus, $\sum_{j=1}^M |a_j| \le 1$. We now set $c_j = a_j$ for all $1 \le j \le M$ and $c_j=0$ otherwise. Note that since $M \le n$, this is well-defined.

  By construction, the polynomial
  $p(z)$ defined as $p(z) = \sum_{j=0}^N c_j z^j$ has at least
  $T$ repeated roots at $1$, where
  $$
  T =  \min \bigg\{\frac{2}{7} \sqrt{M \cdot \ln (1/2 \epsilon)}, M\bigg\} = \frac{2}{7} \sqrt{M \cdot \ln (1/2 \epsilon)},
  $$
  where the last equality uses $1- \err \ge 2 \ln (2/\epsilon)/n$. We note for later reference that
  \begin{equation} \label{eq:Tbound}
  T = \Omega\left(n^{1/3} \cdot \ln^{2/3}(1/\epsilon) \cdot (1-\err^2)^{1/3} \cdot \err^{-2/3}\right).
  \end{equation}
  Let us define $\theta^{\ast}$ as
\begin{equation} \label{eq:thetaastdef}
  \theta^\ast = \frac{2T}{9M} = {\frac 4 {63}} \sqrt{{\frac {\ln(1/2\eps)} M}} \leq \frac{4}{63} \cdot \frac{\ln^{1/3}(1/\epsilon) \cdot \err^{2/3}}{n^{1/3} \cdot (1-\err^2)^{1/3} }.
\end{equation}
  Observe that since $1-\err \ge 2 \ln(1/\epsilon)/n$, it holds that $\theta^\ast \le 4/63$.
  Let $A$ be the arc of the unit circle $A = \{e^{i \theta} | -\theta^\ast \le \theta \le \theta^\ast\}$.
  Applying Lemma~\ref{lem:bound-repeated} (and observing that all degree $M+1$ and higher coefficients of $p$ are zero), we obtain that
  \begin{equation}\label{eq:bound-p}
  \sup_{z \in A} |p(z)| = 9 \cdot (M +1) \cdot \bigg( \frac{e}{9}
 \bigg)^T \le \bigg( \frac{1}{3}
 \bigg)^T. \end{equation}
  Here the last inequality uses $T = \frac{2}{7} \sqrt{M \cdot \ln (1/2 \epsilon)}$ and the fact that $\eps$ is at most some sufficiently small constant.

 \ignore{
 observe that the polynomial $F_c(z)$ can be expressed as
  $$
  F_c(z) = \left(\frac{1+\err}{2} + \frac{1-\err}{2}z\right)^{n} \cdot \bigg(\sum_{i=0}^n c_i \bigg(\frac{\frac{1-\err}{2} + \frac{1+\err}{2}z}{\frac{1+\err}{2} + \frac{1-\err}{2}z} \bigg)^i \bigg).
  $$
  }
Now we turn our attention to $F_c(z).$  Recalling (\ref{eq:Fczw}), we have that
\begin{equation} \label{eq:rephrase}
 \sup_{|z|=1} |F_c(z)| = \sup_{|w|=1}\left| \bigg( \frac{\err}{\frac{1+\err}{2} -\frac{1-\err}{2} w}\bigg)^n   \cdot \sum_{i=0}^n c_i w^i\right|.
 \end{equation}
 Let us write $\Phi_c(w)$ to denote $\left( \frac{\err}{\frac{1+\err}{2} -\frac{1-\err}{2} w}\right)^n   \cdot \sum_{i=0}^n c_i w^i$, so we seek to upper bound $\sup_{|w|=1} |\Phi_c(w)|$.
 We do this by upper bounding $|\Phi_{c}(w)|$ separately on the sets $A$ and $\overline{A}$.

 First, we bound $|\Phi_{c}(w)|$ in the set $A$ as follows:
 \begin{equation}\label{eq:bound-Fc-A}
 \sup_{w \in A} |\Phi_c(w)| \le \sup_{w \in A} |p(w)| \le e^{-\Omega(T)} =  \exp \bigg( - \Omega \bigg(\frac{\ln^{2/3}(1/\epsilon) \cdot (n(1-\err^2))^{1/3}}{\err^{2/3}}\bigg)\bigg).
 \end{equation}
 Here the first inequality uses the fact that $\bigg| \frac{\err}{\frac{1+\err}{2} -\frac{1-\err}{2} w}\bigg| \le 1$, the second inequality uses (\ref{eq:bound-p}), and the last equality uses (\ref{eq:Tbound}).

 To bound $|\Phi_c(w)|$ in $\overline{A}$, we will need a couple of facts. First, since $\sum_{j=0}^n |c_j| \le 2$, it is the case that $|p(w)|\le 2$ for all $|w|=1$, and consequently
 \[
 |\Phi_c(w)| \le 2 \left|\frac{\err}{\frac{1+\err}{2} -\frac{1-\err}{2} w} \right|^n.
 \]
Recalling (\ref{eq:burger}), \ignore{writing $w=e^{i\theta}$ we have that
$$
 \bigg| \frac{\err}{\frac{1+\err}{2} -\frac{1-\err}{2} w}\bigg|^2  = \frac{2\err^2}{(1-\cos \theta) + \err^2 (1 + \cos \theta) } =  \bigg( \frac{1}{1 + \frac{(1-\err^2) \sin^2(\theta/2) }{\err^2}}\bigg).
$$

Thus }we have
 $$
\sup_{w \in \overline{A}} |\Phi_c(w)| \le 2  \bigg( \frac{1}{1 + \frac{(1-\err^2) \sin^2 (\theta^\ast/2) }{\err^2}}\bigg)^{n/2} \le 2  \bigg( \frac{1}{1 + \frac{(1-\err^2)  (\theta^\ast)^2 }{8\err^2}}\bigg)^{n/2},
$$
where the last inequality uses $\sin^2(\theta^\ast/2) \geq (\theta^\ast)^2/8$ which holds since $\theta^\ast \le 4/63$.
Finally, again using
$\err^4 \ge \ln(1/\epsilon)/n$ and recalling (\ref{eq:thetaastdef}), we have
$
\frac{(1-\err^2)  (\theta^\ast)^2 }{8\err^2} \le 4/63
$ (with room to spare).
Thus, we have that
\begin{align*}
 \sup_{w \in \overline{A}} |\Phi_c(w)| &\le  2  \bigg( \frac{1}{1 + \frac{(1-\err^2)  (\theta^\ast)^2 }{8\err^2}}\bigg)^{n/2}
 \le \exp \bigg( -\Omega \bigg(  \frac{(1-\err^2)  (\theta^\ast)^2 n }{\err^2}\bigg) \bigg) \\
& \le \exp \bigg( - \Omega \bigg(\frac{\ln^{2/3}(1/\epsilon) \cdot (n(1-\err^2))^{1/3}}{\err^{2/3}}\bigg)\bigg),
 \end{align*}
where for the last inequality we used $\theta^\ast = \Theta(1) \cdot \frac{\ln^{1/3}(1/\epsilon) \cdot \err^{2/3}}{n^{1/3} \cdot (1-\err^2)^{1/3} }$, which follows from (\ref{eq:thetaastdef}). Combining with (\ref{eq:bound-Fc-A}) finishes the proof.
 \end{proofof}

\subsection*{Acknowledgments}
A.~D.  would like to thank Mike Saks for suggesting the noisy population recovery problem for unrestricted support  and for many illuminating conversations about this problem. The authors would like to thank Tamas Erd\'{e}lyi for several helpful email exchanges about \cite{Erdelyi16}. 
\bibliographystyle{alpha}
\bibliography{allrefs}

\end{document}